\documentclass{stylefile}

\usepackage{amsmath,amssymb,amsthm,enumerate,multirow,subfigure,url}
\usepackage[usenames]{color}

\renewenvironment{proof}{\noindent
\textsc{{\bf Proof.}}}{\bigskip}

\newtheorem{observation}{Observation}
\newtheorem{theorem}{Theorem}
\newtheorem{remark}{Remark}
\newtheorem{example}{Example}

\newcommand{\R}[0]{\mathbb{R}}
\newcommand{\N}[0]{\mathbb{N}}
\newcommand{\E}[0]{{\cal E}}
\newcommand{\ute}{UTE_{\instance}}
\newcommand{\lte}{LTE_{\instance}}
\newcommand{\uts}{UTS_{\instance}}
\newcommand{\lts}{LTS_{\instance}}
\newcommand{\ut}[2]{u_{#1}({#2})}
\newcommand{\nut}[2]{u'_{#1}({#2})}
\newcommand{\utb}[2]{\bar{u}_{#1}({#2})}
\newcommand{\nutb}[2]{\bar{u}'_{#1}({#2})}
\newcommand{\lt}[2]{l_{#1}({#2})}
\newcommand{\nlt}[2]{l'_{#1}({#2})}
\newcommand{\ltb}[2]{\bar{l}_{#1}({#2})}
\newcommand{\nltb}[2]{\bar{l}'_{#1}({#2})}

\newcommand{\instance}[0]{{\cal P}}
\newcommand{\newinstance}[2]{\instance_{c_{{#1},{#2}}}}
\newcommand{\maxdec}[1]{\delta(#1)}
\newcommand{\Dminus}[1]{D_{-}(#1)}
\newcommand{\Dplus}[1]{D_{+}(#1)}
\newcommand{\costname}[0]{f_c}
\newcommand{\costs}[1]{\costname(#1)}
\newcommand{\newcostname}[1]{{f_{c_{#1}}}}
\newcommand{\newcosts}[2]{{f_{c_{#1}}}(#2)}

\begin{document}

\begin{frontmatter}

\title{Extending the definition of single and set tolerances}

\author[Umea]{Gerold J\"ager}
\address[Umea]{Department of Mathematics and
Mathematical Statistics\\
University of Ume{\aa}\\
SE-90187 Ume{\aa}, Sweden\\
{\rm gerold.jager@umu.se}
}

\author[Aarhus]{Marcel Turkensteen}
\address[Aarhus]{Department of Economics and Business Economics\\
School of Business and Social Sciences\\
University of Aarhus\\
Fuglesangs Alle 4\\
DK-8210 Aarhus, Denmark\\
{\rm matu@econ.au.dk}
}

\begin{abstract}
Optimal solutions of combinatorial optimization problems can be sensitive to changes in the cost of one or more
elements of the ground set $\E$.
Single and set tolerances measure the supremum / infimum possible change such that
the current solution remains optimal
for cost changes in one or more elements. The current definition
does not apply to all elements of $\E$ or to all subsets of $\E$. In this work, we
broaden the definition to all elements for single tolerances and
to all subsets of elements for set tolerances, while proving that
key theoretical and computational properties still apply.
\end{abstract}

\end{frontmatter}

\section{Introduction}
\label{sec_intro}

Many studies consider the effect of parameter changes
on the optimality of solutions. Greenberg~\cite{Gre94} performed sensitivity analysis
(there it is  called postoptimal analysis) on linear programming problems.
Andersen et al.~\cite{AKR23} developed a computational approach for determining
under which parameter values Mixed Integer Linear Programming solutions remain optimal,
but if multiple parameters change, it is necessary to solve a number of problems
which is exponential in the number of changing parameter values.

In this work, we consider combinatorial optimization problems (COPs), in which
one selects a combination of elements from the ground set to minimize (maximize)
some objective function based on the costs of the individual elements.
Examples of COPs are given in~\cite{CCPS97} and include decision problems in Binary Programming.
The objective is often to minimize (maximize) the sum of the costs of elements in a solution,
but one can also minimize the maximum (bottleneck) cost, as in~\cite{KZ13}, or the product of
the costs (see~\cite{JT24} for applications).
Further relevant studies on sensitivity analysis for COPs can be found in~\cite{FV18,MWMS22}.

In previous studies~\cite{GJM06b,JT18,JT24}, we developed a theoretical framework
called \emph{theory of tolerances}, in which we determined the limits to the changes
in the costs of elements or sets of elements such current solutions remain optimal and thus other solutions become
optimal. These limits are called \emph{tolerances}.

In case of parameter changes of single elements,
so-called \emph{single tolerances} provide the supremum increase and decrease, respectively, in
which optimal solutions remain optimal.
In case of multiple simultaneous parameter changes, solutions remain optimal within some region
formed by combinations of individual cost changes. However, determining this region is a laborious
task~\cite{AKR23}. The theoretical framework
includes measures of the size of the area, called \emph{set tolerances}, which are also
independent of the optimal solution.
In fact, the theoretical framework shows that
the introduced single and set tolerances are independent
of the choice of the optimal solution. This is beneficial, because cost changes could
make some but not all optimal solutions non-optimal.

The current definitions of single and set tolerances have the drawback that they
only apply to specific  elements or specific subsets of elements of  the ground set. As a consequence, if there is a
set of elements in which all costs either increase or decrease, we first have to check whether the
tolerance is defined, and we may discover that this is not the case.

In this work we generalize the definitions and extend the theory of single and set
tolerances in such a way that this drawback no longer holds. We show that all
extended / new definitions are consistent
with the corresponding current ones.
We further prove that several previous results, in particular exact values and upper and lower
bounds, can be generalized for the new definitions so that the main part of the previous
theory can be kept.

In this work, we aim to provide new definitions of single and
set tolerances, which
\begin{itemize}
\item are consistent with the current definitions on their
respective domains,
\item can be computed conveniently (i.e., previous computational
results should be applicable to all
elements/subsets).
\end{itemize}

\section{Notations, definitions, and existing results}
\label{sec_notate}

For the sake of completeness and readability of this work,
we present notations and the definitions of
single upper/lower tolerances, based on~\cite{GJM06b}, of
regular set upper/lower tolerances, based on~\cite{JT18},
and of reverse set upper/lower tolerances,
based on~\cite{JT24}.

\subsection{Combinatorial minimization problems}

Formally, a {\emph combinatorial minimization problem (CMP)} $\instance$
is given by a tuple
$(\E,D,c,\costname)$ where
$\E$ is a finite \emph{ground set of elements},
$D\subseteq 2^\E \setminus \{ \emptyset \} $ is
the set of \emph{feasible solutions},
$c:\,\E\to \R$ is the \emph{cost function}
which assigns costs to each single element \mbox{of $\E$},
$\costname:\,D \to \R$ is
a function which depends on the function~$c$ and assigns
costs to each feasible solution $D$.

$S^\ast\subseteq \E$ is called an \emph{optimal solution} of
$\instance$ if $S^\ast$ is a feasible solution and the costs
$\costs{S^\ast}$ of $S^\ast$ are infimum.
We denote the cost of an optimal solution
$S^\ast$ of $\instance$ by $\costs{\instance}$
and the set of optimal solutions by $D^\ast$.
Moreover, we require that there is at least one optimal solution of $\instance$, i.\,e., $D^\ast\not=\emptyset$.

We consider CMPs $\instance=(\E,D,c,\costname)$
where
$\costname:\, D\to\R$ is of type
\emph{sum}, i.e., the sum of all costs is minimized and which we call
a \emph{combinatorial sum problem (CSP)},
where it is of type \emph{product}, i.e., the product of all (positive) costs is minimized
and which we call
a \emph{combinatorial product problem (CPP)}, and where it is of type \emph{bottleneck},
i.e., the maximum of all costs is minimized, and which we call
a \emph{combinatorial bottleneck problem (CBP)}.

Define
\[\maxdec{e}:=\left\{\begin{array}{ll}
\infty, & \mbox{if $ \costname $ is of type sum or bottleneck}
\\
c(e), & \mbox{if $ \costname $ is of type product}
\end{array}\right.\]
as the supremum by which element $e$ can be
decreased such that $\costname$ remains of type
sum, product, or bottleneck, respectively.
For a CBP, let $c^\ast$ denote
the objective value of an optimal
solution, so $c^\ast = \costs{\instance}$.

Let a CMP $\instance=(\E,D,c,\costname)$
and $ E = \{e_1,e_2,\dots,e_k\} \subseteq \E $ be given.
If for all $ l \in \N $ with
$ 1 \le l \le k $ we add $ \alpha_l \in \R $ to the cost of $e_l$,
we denote a new CMP instance by
$\newinstance{\vec{\alpha}}{E}=(\E,D,c_{\vec{\alpha},E},\newcostname{\vec{\alpha},E})$,
with costs
$ c_{\vec{\alpha},e}\bar{e}) = \left\{
\begin{array}{cl}
c(\bar{e}), & \mbox{if}\ \bar{e} \in \E \setminus E
\\
c(\bar{e})+\alpha_l, & \mbox{if}\ \bar{e}=e_l\ \mbox{for}\ l=1,2,\dots,k
\end{array}\right. $.

Note that
$\newcostname{\vec{\alpha},E}$ is of the same type as
$\costname$,
unless $\costname$ is of type product and $ \alpha_l \le
-c(e_l) $ for at least one $ l \in \{1,2,\dots,k\} $.
When $|E|=1$ with $E=\{ e\}$, we write
$\newinstance{\alpha}{e}=(\E,D,c_{\alpha,e},\newcostname{\alpha,e})$,
for the sake of simplicity, instead of the formally correct notation with $\{ e\}$.

For $M\subseteq D$, we denote the costs of the
best solution included in $M$ by $\costs{M}$.
Furthermore, the cost $\costs{\emptyset}$ is defined as infinite, i.\,e., $\infty$.

Let $e \in \E$. We define the subsets of solutions
$ \Dplus{e}=\{\,S\in D \mid  e\in S \} $ and $ \Dminus{e}=\{\,S\in D \mid e \in \E \setminus S \} $.

In Sections~\ref{ssec_siupp}--\ref{ssec_setlow}, the current
definitions of upper/lower single tolerances~\cite{GJM06b},
of regular upper/lower set tolerances~\cite{JT18},
of reverse upper/lower set tolerances~\cite{JT24} are given. All these definitions depend on a given optimal solution $ S^\ast $, where
$ e \in S^\ast $ for single upper tolerances,
$ e \in \E \setminus S^\ast $ for single lower tolerances,
$ E \subseteq S^\ast $ for set upper tolerances,
$ E \subseteq \E \setminus S^\ast $ for set lower tolerances.
However, they are shown to be independent of $ S^\ast $ in these references.
The following observation is the basis for a convenient reformulation
of these definitions.

\begin{observation}
\label{obs_settol}
Let $ \instance =(\E,D,c,\costname)$
be a CMP instance and $ E =\{e_1,e_2,\dots,e_k\} \subseteq \E $.
Furthermore, let $\vec{\alpha} = (\alpha_1,\alpha_2,\dots,\alpha_k) $,
where $ \alpha_l \ge 0 $ for all $ l=1,2,\dots,k$.
\begin{description}
\item[{\bf (a)}]
Let $S_1^\ast$ and $S_2^\ast$ be optimal solutions of $\instance$ with
$E \subseteq S_1^\ast$ and $E \nsubseteq S_2^\ast$.
Then it holds that
$\newcostname{\vec{\alpha},E}(S_2^\ast) \leq \newcostname{\vec{\alpha},E}(S_1^\ast) $.
\item[{\bf (b)}]
Let $S_1^\ast$ and $S_2^\ast$ be optimal solutions of $\instance$ with
$E \subseteq \E \setminus S_1^\ast$ and $E \nsubseteq \E \setminus S_2^\ast$.
Then it holds that
$\newcostname{-\vec{\alpha},E}(S_2^\ast) \leq \newcostname{-\vec{\alpha},E}(S_1^\ast) $.
\end{description}
\end{observation}

\subsection{Single upper tolerances}
\label{ssec_siupp}

Let $ \instance =(\E,D,c,\costname)$
be a CMP instance and let $ \ute := \bigcup_{S^\ast\in D^\ast} S^\ast$.
We define the \emph{(single) upper tolerance} $\ut{\instance}{e}$
of each $e\in \ute$ as follows:
\begin{eqnarray}
\label{newdef_sin_upp}
\nonumber
& & \ut{\instance}{e}
\\
\label{ut_sin_def}
& := &
\sup \; \{
\alpha\in\R_0^+
\mid
\mbox{each opt.\ $ S^\ast $ for}\ \instance\
\mbox{is opt.\ for}\
\newinstance{\alpha}{e} \}
\\
\nonumber
& := & \inf \; \{\alpha\in\R_0^+\mid
\mbox{some opt.\ $ S^\ast $ for}\ \instance\
\mbox{is not opt.\ for}\
\newinstance{-\alpha}{e} \}.
\end{eqnarray}
In the case of $ \ut{\instance}{e} < \infty $, ``sup'' can be replaced by ``max''.

This definition is a slight modification from the one provided in~\cite{GJM06b}.
As optimal solutions $ S^\ast$ with $ e \in \E \setminus S^\ast$ remain automatically
optimal when increasing $ c(e)$, both definitions are equivalent.

We can determine the set $\ute$ from all single lower tolerances if $\costname$
is of type sum or product~\cite[Theorem 12]{GJM06b} and from all so-called
\emph{smallest includes} if $\costname$ is of type bottleneck ~\cite[Algorithm
1]{TJ22}.

\subsection{Single lower tolerances}
\label{ssec_silow}

Let $ \instance =(\E,D,c,\costname)$
be a CMP instance and let $ \lte := \E \setminus \bigcap_{S^\ast\in D^\ast} S^\ast$.
We define the \emph{(single) lower tolerance}
of each $e\in \lte$ as follows:
\begin{eqnarray*}
& & \lt{\instance}{e}
\\
& := & \sup \; \{\alpha\in\R_0^+\mid
\mbox{each opt.\ $ S^\ast $ for}\ \instance\
\mbox{is opt.\ for}\
\newinstance{-\alpha}{e} \}\\
& := & \inf \; \{\alpha\in\R_0^+\mid
\mbox{some opt.\ $ S^\ast $ for}\ \instance\
\mbox{is not opt.\ for}\
\newinstance{-\alpha}{e} \}.
\end{eqnarray*}

In the case of $ \lt{\instance}{e} < \infty $, ``sup'' can be replaced by ``max''.

This definition is a slight modification from the one provided in~\cite{GJM06b}.
By Observation~\ref{obs_settol}(b), applied to sets of cardinality one, it holds
for a decrease of $ c(e)$ by $\alpha \ge 0 $ that
if any optimal solution $ S^\ast$ with $ e \in \E \setminus S^\ast$ remains
optimal, then also any optimal solution $ S^\ast$ with $ e \in S^\ast$ remains
optimal.
Thus, both definitions are equivalent.

We can determine the set $\lte$ from all single upper tolerances if $\costname$
is of type sum or product~\cite[Theorem 6]{GJM06b}, but we need both all upper
tolerances and smallest includes if $\costname$ is of type bottleneck
~\cite[Algorithm 1]{TJ22}.

\subsection{Regular and reverse set upper tolerances}
\label{ssec_setupp}

Let $ \instance =(\E,D,c,\costname)$
be a CMP instance and let
$ \uts := \{ E \subseteq \E \mid \exists \, S^\ast \in D^\ast:
\, E \subseteq S^\ast \} $.
We define the {\em regular set upper tolerance}
for each
$ E = \{e_1,e_2,\dots,e_k \} \in \uts $ as follows:
\begin{eqnarray}
\nonumber
& & \ut{\instance}{E}
\\
\label{newdef_reg_upp}
& := &
\sup \; \Big\{
\alpha\in\R
\mathrel{\Big|}
\mbox{for each optimal $ S^\ast $ for}\ \instance\
\\
\nonumber
&  & \exists \, \vec{\alpha} = (\alpha_1,\alpha_2,\dots,\alpha_k),
\alpha_1, \alpha_2,\dots, \alpha_k \ge 0,\
\\
& &
\nonumber
\alpha = \sum_{l=1}^k \alpha_l, \;
S^\ast\ \mbox{is optimal for}\
\newinstance{\vec{\alpha}}{E} \Big\}.
\end{eqnarray}

In the case of $ \ut{\instance}{E} < \infty $, ``sup'' can be replaced by ``max'',

This definition is a slight modification from the one provided in~\cite{JT18}.
By Observation~\ref{obs_settol}(a) it holds
for an increase of $ c(e_s)$ by $\alpha_l \ge 0 $ for $l=1,2,\dots,k$
that if an optimal solution $ S^\ast$ with $ E \subseteq S^\ast$ remains optimal,
then also optimal solutions $ S^\ast$ with $ E \nsubseteq S^\ast$ remain optimal.
Thus, both definitions are equivalent.

We define the {\em reverse set upper tolerance}
for each
$ E = \{e_1,e_2,\dots,e_k \} \in \uts $ as follows:
\begin{eqnarray}
\nonumber
& & \utb{\instance}{E}
\\
& := &
\label{ut_revset_def}
\inf \; \Big\{
\alpha\in\R
\mathrel{\Big|}
\exists \, \mbox{optimal $ S^\ast $ for}\ \instance\
\\
& &
\nonumber
\exists \,  \vec{\alpha} = (\alpha_1,\alpha_2,\dots,\alpha_k),
\; \alpha_1, \alpha_2,\dots, \alpha_k \ge 0,\
\\
& &
\nonumber
\alpha = \sum_{l=1}^k \alpha_l, \;
S^\ast\ \mbox{is not optimal for}\
\newinstance{\vec{\alpha}}{E} \Big\}.
\end{eqnarray}

This definition is a slight modification from the one provided in~\cite{JT24}.
By Observation~\ref{obs_settol}(a), it holds
for an increase of $ c(e_s)$ by $\alpha_l \ge 0 $ for $l=1,2,\dots,k$ that
if there is an optimal solution $ S^\ast$ which loses optimality,
then also an optimal solution $ S^\ast $ with $ E \subseteq S^\ast$ loses optimality.
Thus, both definitions are equivalent.

By definition, it holds that
$ \utb{\instance}{E} \le \ut{\instance}{E} $
and $ \ut{\instance}{\{e\}} = \utb{\instance}{\{e\}} \; = \;
\ut{\instance}{e}$.

\subsection{Regular and reverse set lower tolerances}
\label{ssec_setlow}

Let $ \instance =(\E,D,c,\costname)$
be a CMP instance and let
$ \lts := \{ E \subseteq \E \; \Big| \; \exists \, S^\ast
\in D^\ast: \, E \subseteq \E \setminus S^\ast \} $.
We define the {\em regular set lower tolerance}
for each $ E = \{e_1,e_2,\dots,e_k \}
\in \lts $ as follows:
\begin{eqnarray*}
\label{newdef_reg_low}
\lt{\instance}{E} & := &
\sup \; \Big\{ \alpha\in\R
\mathrel{\Big|}
\mbox{for each optimal $ S^\ast $ for} \ \instance
\\
& &
\exists \, \vec{\alpha} = (\alpha_1,\alpha_2,\dots,\alpha_k),
0 \le \alpha_1 < \maxdec{e_1},
\\
& &
0 \le \alpha_2 <  \maxdec{e_2}, \dots,
0 \le \alpha_k < \maxdec{e_k},\
\\
& &
\alpha = \sum_{l=1}^k \alpha_l, \;
S^\ast\ \mbox{is optimal for}\
\newinstance{-\vec{\alpha}}{E} \Big\}.
\end{eqnarray*}
In the case of $ \lt{\instance}{E} < \infty $, ``sup'' can be replaced by ``max''.

This definition is a slight modification from the one provided in~\cite{JT18}.
By Observation~\ref{obs_settol}(b), it holds
for a decrease of $ c(e_s)$ by $\alpha_l \ge 0 $ for $l=1,2,\dots,k$ that if
an optimal solution $ S^\ast$ with $ E \subseteq \E \setminus S^\ast$ remains optimal,
then also optimal solutions $ S^\ast$ with $ E \nsubseteq \E \setminus S^\ast$ remain optimal.
Thus, both definitions are equivalent.

We define the {\em reverse set lower tolerance}
for each
$ E = \{e_1,e_2,\dots,e_k \} \in \lts $ as follows:
\begin{eqnarray*}
\label{newdef_rev_low}
\ltb{\instance}{E} & := &
\inf \; \Big\{ \alpha\in\R
\mathrel{\Big|}
\exists \, \mbox{optimal $ S^\ast $ for}\ \instance\
\\
& &
\exists \, \vec{\alpha} = (\alpha_1,\alpha_2,\dots,\alpha_k),
\; 0 \le \alpha_1 < \maxdec{e_1},
\\
& &
0 \le \alpha_2 <  \maxdec{e_2}, \dots,
0 \le \alpha_k
< \maxdec{e_k},\
\\
& &
\alpha = \sum_{l=1}^k \alpha_l, \;
S^\ast\ \mbox{is not optimal for}\
\newinstance{-\vec{\alpha}}{E} \Big\}.
\end{eqnarray*}

This definition is a slight modification from the one provided in~\cite{JT24}.
By Observation~\ref{obs_settol}(b), it holds
for a decrease of $ c(e_l)$ by $\alpha_l \ge 0 $ for $l=1,2,\dots,k$ that
if there is an optimal solution $ S^\ast$ which loses optimality,
then also an optimal solution $S^\ast$ with $ E \subseteq \E \setminus S^\ast$
loses optimality.
Thus, both definitions are equivalent.

By definition, it holds that:
$ \ltb{\instance}{E} \; \le \; \lt{\instance}{E} $
and $ \lt{\instance}{\{e\}} \; = \; \ltb{\instance}{\{e\}} \; = \; \lt{\instance}{e} $.

\subsection{Interpretation of regular and reverse set tolerances}
\label{sec_interpret}

Consider a set of elements $E$ for which all element costs can either increase (set upper tolerances)
or decrease (set lower tolerances).
There are certain combinations of cost changes for which all optimal solutions $S^\ast$ remain optimal,
the optimality region (illustrated graphically for $|E|=2$ in ~\cite{JT24}). Note that the region in which
all solutions remain optimal is the intersection of the optimality regions of all optimal solutions.
The regular set tolerance is the $1$-norm distance to any point in this region farthest away from the origin,
whereas the reverse set lower tolerance is the (limit of the) $1$-norm distance to any closest point
outside the region.

\section{Examples}

The following examples illustrate the different tolerance definitions as well as their
benefits and shortcomings.
Note that not all single and set tolerances are currently defined in these examples.

\begin{example}
\label{examp0}
Consider the following CSP.
\begin{description}
\item[$\bullet$] $\E=\{w,x,y,z\}$ with $c(w)=3$, $c(x)=5$, $c(y)=4$, $c(z)=6$,
\item[$\bullet$] $D=\{\,\{w,x\},\, \{y,z\}\}$,
\item[$\bullet$]
$ S_1 := \{w,x\} $ is optimal with cost $8$, whereas
$S_2 := \{y,z\} $ is non-optimal with cost $10$.
\end{description}
\end{example}

\begin{example}
\label{examp1}
Consider the following CSP.
\begin{description}
\item[$\bullet$] $\E=\{v,w,x,y,z\}$ with $c(v)=2$, $c(w)=3$, $c(x)=5$, $c(y)=4$, $c(z)=8$,
\item[$\bullet$] $D=\{\,\{v,x\},\, \{w,y\},\,\{z\}\}$,
\item[$\bullet$]
$ S_1 := \{v,x\} $ and $S_2:=\{w,y\}$ are optimal with cost $7$, whereas
$S_3 := \{z\} $ is non-optimal with cost $8$.
\end{description}

\pagebreak[3]

\end{example}
\begin{example}
\label{examp2}
Consider the following CPP.
\begin{description}
\item[$\bullet$] $\E=\{v,w,x,y,z\}$ with $c(v)=2$, $c(w)=3$, $c(x)=6$, $c(y)=4$,
$c(z)=24$,
\item[$\bullet$] $D=\{\,\{v,x\},\, \{w,y\},\,\{z\}\}$,
\item[$\bullet$]
$ S_1 := \{v,x\} $ and $S_2:=\{w,y\}$ are optimal with cost $12$, whereas
$S_3 := \{z\} $ is non-optimal with cost $24$.
\end{description}
\end{example}

\pagebreak[3]

\begin{example}
\label{examp3}
Consider the following CBP.
\hspace*{0em}
\begin{description}
\item[$\bullet$] $\E=\{w,x,y,z\}$ with $c(w)=2$, $c(x)=3$, $c(y)=7$, $c(z)=8$.
\item[$\bullet$] $D=\{\,\{w,y\},\, \{x,y\},\,\{z\}\}$,
\item[$\bullet$]
$ S_1 := \{w,y\} $ and $S_2:=\{x,y\}$ are optimal with cost $7$, whereas
$S_3 := \{z\} $ is non-optimal with cost $8$.
\end{description}
\end{example}

\section{Generalized definitions of upper and lower tolerances}
\label{sec_newdef}

\subsection{New and extended definitions of tolerances}
\label{ssec_newext}

In this section, we wish to provide definitions of single and set tolerances to all
elements of $\E$ and to all subsets of $\E$, respectively.
These definitions should be \emph{consistent} with the current definition for the
domains on which single and set tolerances are defined. Moreover,
tolerances computations should be as \emph{convenient} as possible.

\subsubsection*{Upper tolerances}

For upper tolerances, we choose to extend the current definition
of single upper tolerances~\eqref{ut_sin_def} to any $e\in\E$
and the current definitions of
regular set upper tolerances~\eqref{newdef_reg_upp}
and reverse set upper tolerances~\eqref{ut_revset_def}
to any $E\subseteq \E$, which
we call \emph{extended definitions}.
We use the prime to distinguish the extended definitions
and use the notation $\nut{\instance}{e}$ for single upper tolerances
applied to the element $e$
and the notation $\nut{\instance}{E}$ / $\nutb{\instance}{E}$
for regular / reverse set upper tolerances
applied to the subset $E$.

\subsubsection*{Lower tolerances}

For lower tolerances, we introduce \emph{new definitions},
which for regular / reverse lower tolerances are as follows:
\begin{eqnarray}
\label{newdef_reg_lowb}
\nlt{\instance}{E} & := &
\sup \; \Big\{ \alpha\in\R
\mathrel{\Big|}
\exists \, \vec{\alpha} = (\alpha_1,\alpha_2,\dots,\alpha_k),
\\
\nonumber
& &
0 \le \alpha_1 < \maxdec{e_1},
0 \le \alpha_2 <  \maxdec{e_2}, \dots,
0 \le \alpha_k
\\
& &
\nonumber
< \maxdec{e_k},\
\alpha = \sum_{l=1}^k \alpha_l, \;
\newcosts{-\vec{\alpha},E}{\newinstance{-\vec{\alpha}}{E}} = \costs{\instance} \}.
\end{eqnarray}

In the case of $ \nlt{\instance}{E} < \infty $, ``sup'' can be replaced by ``max''.
\begin{eqnarray}
\nonumber
& & \nltb{\instance}{E}
\\
\label{newdef_rev_lowb}
& := &
\inf \; \Big\{ \alpha\in\R
\mathrel{\Big|}
\exists \, \vec{\alpha} = (\alpha_1,\alpha_2,\dots,\alpha_k),
\\
& &
\nonumber
\; 0 \le \alpha_1 < \maxdec{e_1},
0 \le \alpha_2 <  \maxdec{e_2}, \dots,
0 \le \alpha_k
\\
& &
\nonumber
< \maxdec{e_k},\
\alpha = \sum_{l=1}^k \alpha_l, \;
\newcosts{-\vec{\alpha},E}{\newinstance{-\vec{\alpha}}{E}} < \costs{\instance} \}.
\end{eqnarray}

The equivalent definitions of single lower tolerances are then:
\begin{eqnarray}
\label{newdef_sin_low}
\nlt{\instance}{e}
& := & \sup \; \{\alpha\in\R\mid
\newcosts{-\alpha,e}{\newinstance{-\alpha}{e}}
= \costs{\instance} \}.
\\
\nonumber
& := & \inf \; \{\alpha\in\R\mid
\newcosts{-\alpha,e}{\newinstance{-\alpha}{e}}
< \costs{\instance} \}.
\end{eqnarray}

\subsection{Concepts and motivation}
\label{ssec_conc}

Now, we explain why our definitions of upper tolerances are simply extensions
to the current definitions, whereas our definitions of lower tolerances are new concepts. Moreover,
we explain why one should not adopt new definitions for upper tolerances or use
extended definitions for lower tolerances.
Our new definitions of upper tolerances are extensions of the current definitions, whereas
for the new lower tolerance definitions we consider the change in the objective value.

\subsubsection*{Upper tolerances}

We choose the extended definitions for upper tolerances,
because they are consistent with the current definitions,
while the new definitions are not.
For example, the upper tolerances of $v$, $w$, $x$, $y$ in Example~\ref{examp1} are $0$ under the current and the
extended definition, whereas the upper tolerances of \emph{all} elements $v$, $w$, $x$, $y$, $z$
in Example~\ref{examp1} would be infinite
 under the new definition, as the objective value remains the same if the costs of each of these elements
increased by any positive amount.
This also illustrates that the interpretation of the resulting values can become problematic.

\subsubsection*{Lower tolerances}

We show that the lower tolerances under the new definitions are consistent with those
according to the current definitions for single lower tolerances of any $e\in\lte$ and for set tolerances
of any $E\subseteq \lts$.
Moreover, we show that computation formulas and bounds can be extended to any element in or any subset of $\E$ without the need for
$\lte$ or $\lts$.

A further benefit of the new definition of single lower tolerances is that it is consistent with the
computation of lower tolerances of all elements both for the $1$-Tree Problem in~\cite{Hel00}
(called \emph{nearness values}) and for the Linear Assignment Problem~\cite{kindervater85}
(called \emph{hard dual values}). 

In contrast, applying the extended definitions of lower tolerances reduces the computational
convenience, in particular for CBPs. For some $e\in \E \setminus \lte $,
the single lower tolerance of $e$
is equal to the infimum decrease in $c(e)$ over all $S$
such that some $S^\ast \in D^\ast $ becomes non-optimal. In Example~\ref{examp3}, we first need
to determine $\E \setminus \lte$ to find that it consists of $y$. This element is included in the optimal solutions
$S_1$ and $S_2$, so we should determine the infimum decrease in $c(y)$ for which either $S_1$ and $S_2$
becomes non-optimal, which in this case would be $4$ (a similar reasoning applies to the reverse set
lower tolerance).
So the extended definitions are less computationally convenient than the new ones.

\section{Computation, bounds, and properties of the extended upper tolerance definitions}
\label{sec_uppset}

In this section, we determine whether the computation of single and set upper tolerances
for $\ute$ and $\uts$, respectively, from previous studies
still applies to the extended definitions in this
work.

The following easily provable remark forms the basis for our results.

\begin{remark}
\label{rem_newdef}
Let $ \instance =(\E,D,c,\costname)$ be a CMP instance
and $ E \subseteq \E $.
Then the following holds:
\begin{description}
\item[{\bf (a)}] $ \nut{\instance}{\bar{E}} \le \nut{\instance}{E} $ for
each $ \bar{E} \subseteq E $.
\item[{\bf (b)}] $ \nutb{\instance}{E} \le \nutb{\instance}{\bar{E}} $ for
each $ \bar{E} \subseteq E $.
\end{description}
\end{remark}

Parts (b), (c), (d) of the following remark show
that the upper tolerance computation formulas in \cite[Theorem 4]{GJM06b}
no longer apply for the extended definition of the
single upper tolerance for arbitrary $ e \in \E $.

\begin{remark}
\label{usi_comp2}
Let $ e \in \E \setminus \ute $.
\begin{description}
\item[\bf (a)]
$ \nut{\instance}{e} = \infty $ holds.
\item[\bf (b)]
Let $\costname$ be of type sum.
Then $\nut{\instance}{e}=\costs{\Dminus{e}}-\costs{\instance}$ does not hold in general.
\item[\bf (c)]
Let $\costname$ be of type product.
Then $\nut{\instance}{e}= \frac{\costs{\Dminus{e}}-\costs{\instance}}{\costs{\instance}}\cdot c(e)$
does not hold in general.
\item[\bf (d)]
Let $\costname$ be of type bottleneck.
Then $\nut{\instance}{e}=\costs{\Dminus{e}}-c(e)$ does not hold in general.
\end{description}
\end{remark}

\begin{proof}
Because of $ e \in \E \setminus \ute $, $ e $ lies outside each optimal solution.
\begin{description}
\item[\bf (a)]
Each optimal solution remains optimal if $ c(e) $
is increased by an arbitrary $ \alpha > 0 $.
It follows that
$ \nut{\instance}{e} = \infty $.
\item[\bf (b),(c),(d)]
By (a), the formulas of (b), (c), (d)
only hold for $ \E \setminus \ute $ if the right-hand terms
are $ \infty$. However,
each of these terms can only be $ \infty$ if the term $ \costs{\Dminus{e}}$
equals $ \infty$, and this is only the case
if each feasible solution contains~$e$.
In other words, the formulas of (b), (c), (d) do not hold in general.
\hfill $\qed$
\end{description}
\end{proof}

The following theorem shows that the bounds and exact values of regular
set upper tolerances~\cite[Theorem 9]{JT18}
can only partly be generalized to arbitrary $ E \subseteq \E $
for the extended definition of the regular set upper tolerance.

\begin{theorem}
\label{urg_comp2}
Let $ \instance =(\E,D,c,\costname)$ be a CMP instance
and $ E = \{e_1,e_2,\dots,e_k\} \subseteq \E $.
Then the following holds:
\begin{description}
\item[{\bf (a)}]
$ \max_{l=1}^k \, \{ \nut{\instance}{e_l} \} \; \le \; \nut{\instance}{E} $.
\item[{\bf (b)}]
If $\costname$ is of type sum or product,
$ \nut{\instance}{E} \; \le \;
\sum_{l=1}^k \nut{\instance}{e_l} $ does \emph{not}
hold in general.
\item[{\bf (c)}]
If $\costname$ is of type bottleneck,
$ \sum_{l=1}^k \nut{\instance}{e_l}
\; \le \;
\nut{\instance}{E} $.
\end{description}
\end{theorem}

\begin{proof}
\begin{description}
\item[\bf (a)]
This follows directly from Remark~\ref{rem_newdef}(a).
\item[\bf (b)]
We distinguish between two cases.
\begin{description}
\item[(A)] $\costname$ is of type sum.

Consider Example~\ref{examp1}.
Let $E=\{v,w\}$, $ e_1 = v$, $e_2 = w $.
It holds that $ \nut{\instance}{v} =  \nut{\instance}{w}  = 0 $
and $ \nut{\instance}{E} = 2 $ as both $S_1$ and $S_2$ remain
optimal, if $ c(v) $ and $c (w) $ are increased by $1$ at the same time.
Thus, $ \nut{\instance}{\{v,w\}} = 2 > 0
= \nut{\instance}{v} + \nut{\instance}{w} $ holds.

\item[(B)] $\costname$ is of type product.

Consider Example~\ref{examp2}.
Let $E=\{v,w\}$, $ e_1 = v$, $e_2 = w $.
It holds that $ \nut{\instance}{v} =  \nut{\instance}{w}  = 0 $
and $ \nut{\instance}{E} = 5 $ as both $S_1$ and $S_2$ remain
optimal, if $ c(v) $ is increased by $2$
and $ c(w) $ is increased by $3$ at the same time.
Thus, $ \nut{\instance}{\{v,w\}} = 5 > 0 =
\nut{\instance}{v} + \nut{\instance}{w} $ holds.

\end{description}

\item[\bf (c)]
Trivially, the inequality is true if $ \nut{\instance}{E} = \infty $.
If $ \sum_{l=1}^k \nut{\instance}{e_l} = \infty $ holds, then
it also holds that $ \max_{l=1}^k \, \{ \nut{\instance}{e_l} \} = \infty $.
It follows from (a) that $ \nut{\instance}{E} = \infty $, and the inequality is
also true. In the following let both terms
$ \nut{\instance}{E} $ and $ \sum_{l=1}^k \nut{\instance}{e_l} $
be not equal $ \infty $.

Let $S^\ast$ be an arbitrary optimal solution of $\instance$.
Let $ \vec{\alpha} = (\alpha_1,\alpha_2,\dots,$ $ \alpha_k) $ and
$ \alpha_l = \ut{\instance}{e_l} < \infty $ for $ l = 1,2,\dots,k $.
We distinguish between two cases.
\begin{description}
\item[(A)] $ E \cap S^\ast = \emptyset$.

Then it holds for each feasible solution $S$:
\begin{align*}
\newcostname{\vec{\alpha},E}(S^\ast) \; = \;
\costs{S^\ast}
\; \le \;
\costs{S}
\; \le \;
\newcostname{\vec{\alpha},E}(S).
\end{align*}

\item[(B)] $ E \cap S^\ast \not= \emptyset$.

Choose $ t \in \{1,2,\dots,k\} $ such that
\begin{eqnarray}
\label{max_rel}
c(e_t) + \alpha_t & = & \max_{l=1,e_l \in S^\ast}^k \{c(e_l) + \alpha_l\}.
\end{eqnarray}
Note that $t$ exists because of
$ E \cap S^\ast \not= \emptyset$ and
$ \alpha_l = \ut{\instance}{e_l} < \infty $ for $ l = 1,2,\dots,k $.
Then it holds for each feasible solution $S$:
\begin{align*}
\newcostname{\vec{\alpha},E}(S^\ast) & \; = \;
\newcostname{\alpha_t,e_t}(S^\ast)
&& \mbox{because of Eq.~\eqref{max_rel}},
\\
& \; \le \; \newcostname{\alpha_t,e_t}(S),
&& \mbox{because of $ \alpha_t = \ut{\instance}{e_t} $},
\\
& \; \le \; \newcostname{\vec{\alpha},E}(S).
\end{align*}

\end{description}

It follows that
$ \sum_{l=1}^k \ut{\instance}{e_l} \le \ut{\instance}{E} $.
\hfill $\qed$
\end{description}
\end{proof}

\begin{remark}
The bounds in Theorem~\ref{urg_comp2}(b) hold for any $E \in \uts$ and for
any $E \subseteq \ute$ with $E \cap (\E \setminus \ute)
\neq \emptyset$.
\end{remark}

\begin{proof}
The first part directly follows from~\cite[Theorem 9(b)]{JT18}. The second
part holds, since the regular set upper tolerance is always infinite when
it contains an element from $(\E \setminus \ute)$, and then the inequality trivially holds.
\hfill \qed
\end{proof}

\smallskip

The following theorem shows that the bounds and exact values of reverse set upper tolerances
~\cite[Theorem 3,4]{JT24} can be generalized to
arbitrary $ E \subseteq \E $.

\begin{theorem}
\label{urv_comp2}
Let $ \instance =(\E,D,c,\costname)$ be a CMP instance
and $ E = \{e_1,e_2,\dots,e_k\} \subseteq \E $.
Then the following holds:
\begin{description}
\item[{\bf (a)}]
$ \nutb{\instance}{E} \le \; \min_{l=1}^k \, \{ \nut{\instance}{e_l} \} $.
\item[{\bf (b)}]
If $\costname$ is of type sum or bottleneck,
$ \nutb{\instance}{E} \; = \;
\min_{l=1}^k \{ \nut{\instance}{e_l} \} $.
\item[{\bf (c)}]
If $\costname$ is of type product,
$ \left(\sqrt[k]{ \min_{l=1}^k
\left\{ \frac{\nut{\instance}{e_l} }{c(e_l)} \right\} +1 } -1 \right)
\cdot \min_{l=1}^k \, \{ c(e_l) \}
\; \le \; \nutb{\instance}{E} $.
\end{description}
\end{theorem}

\begin{proof}
Let $S^\ast$ be an optimal solution corresponding to the definition
of $ \nutb{\instance}{E} $.

\begin{description}
\item[\bf (a)]
This follows directly from Remark~\ref{rem_newdef}(b).
\item[\bf (b)]
By (a), it remains to be shown that
$ \nutb{\instance}{E} \ge \; \min_{l=1}^k \, \{ \nutb{\instance}{e_l} \} $.

To make $ S^\ast $ non-optimal by increasing
some of the costs $ c(e_i) $ for $ i=1,2,\dots,k$ with infimum sum of increases,
we can restrict
to the case that only those $ c(e_i) $ are increased such that
$ e_i \in S^\ast $ and the other $ c(e_i) $ are not increased.
Then for $ \bar{E} := E \cap S^\ast $ it holds that $ \bar{E}
\in \uts $.
By~\cite[Theorem 3(b)]{JT24}
it follows:
\begin{eqnarray*}
\nutb{\instance}{E} \; = \;
\nutb{\instance}{\bar{E}} \; = \;
\min_{l=1,e_l \in \bar{E}}^k \{ \nut{\instance}{e_l} \} \; \ge \;
\min_{l=1}^k \{ \nut{\instance}{e_l} \}.
\end{eqnarray*}
\item[\bf (c)]
Analogously to (b), it holds for $ \bar{E} := E \cap S^\ast $ that $ \bar{E}
\in \uts $.
By~\cite[Theorem 4]{JT24} it follows:
\begin{eqnarray*}
& & \nutb{\instance}{E} \; = \;
\nutb{\instance}{\bar{E}}
\\
& \ge &
\left(\sqrt[k]{ \min_{l=1,e_l \in \bar{E}}^k
\left\{ \frac{\ut{\instance}{e_l} }{c(e_l)} \right\} +1 } -1 \right)
\cdot \min_{l=1, e_l \in \bar{E}}^k \, \{ c(e_l) \}
\\
& \ge &
\left(\sqrt[k]{ \min_{l=1}^k
\left\{ \frac{\ut{\instance}{e_l} }{c(e_l)} \right\} +1 } -1 \right)
\cdot \min_{l=1}^k \, \{ c(e_l) \}.
\hfill \qed
\end{eqnarray*}
\end{description}
\end{proof}

\subsubsection*{Summary}

The computation formulas of single upper tolerances cannot be extended to $\E$: we need
$\ute$ to determine whether the existing computation formulas should be used or whether the value
is infinite (see Remark~\ref{usi_comp2}(a)). The existing bounds still apply to reverse set upper
tolerances (see Theorem~\ref{urv_comp2}) but only partly to regular set
upper tolerances (see Theorem~\ref{urg_comp2}).

\section{Computation, bounds, and properties of the new lower tolerance definitions}
\label{sec_lowset}

In this section, we first show that the new definitions of lower tolerances are consistent with
the current ones for the domains on which the current definitions apply. Subsequently, we show that the
computation formulas for the current definitions apply to the new definitions for all elements in and for all subsets
of $\E$.

\smallskip

The following easily provable remark forms the basis for our results.

\begin{remark}
\label{rem_newdef2}
Let $ \instance =(\E,D,c,\costname)$ be a CMP instance
and $ E \subseteq \E $.
Then the following holds:
\begin{description}
\item[{\bf (a)}] $ \nlt{\instance}{\bar{E}} \le \nlt{\instance}{E} $ for
each $ \bar{E} \subseteq E $.
\item[{\bf (b)}] $ \nltb{\instance}{E} \le \nltb{\instance}{\bar{E}} $ for
each $ \bar{E} \subseteq E $.
\end{description}
\end{remark}

First, we show that the new lower tolerance definitions are consistent
with the current ones.

\begin{theorem}
\label{thm_cons}
Let $ \instance =(\E,D,c,\costname)$ be a CMP instance,
where $\costname $ is of type
sum, product, or bottleneck.
\begin{description}
\item[\bf (a)]
$ \nlt{\instance}{e} = \lt{\instance}{e} $ holds
for each $ e \in \lte $.
\item[\bf (b)]
$ \nlt{\instance}{E} = \lt{\instance}{E} $ holds for each $ E \in \lts$.

\item[\bf (c)]
$ \nltb{\instance}{E} = \ltb{\instance}{E} $ holds for each $ E \in \lts $.
\end{description}
\end{theorem}

\begin{proof}
\begin{description}
\item[\bf (b), (c)]
Let $ E =\{e_1,e_2,\dots,e_k\} \in \lts $,
i.e., there exists an optimal solution $ S^\ast $ with
$E \subseteq \E \setminus S^\ast $.
Let $ \vec{\alpha} = (\alpha_1,\alpha_2,\dots,\alpha_k) $
with $ 0 \le \alpha_l < \maxdec{e_l} $ for $ l=1,2,\dots,k$.

The following statements are equivalent,
where the second statement follows from the first one by Observation~\ref{obs_settol}(b).

\begin{itemize}
\item All solutions $ S^\ast $ with $ E \subseteq \E \setminus S^\ast $
which are optimal for $ \instance $ are also optimal
for $ \newinstance{-\vec{\alpha}}{E}$.
\item All solutions which are optimal for $ \instance $ are also optimal
for $ \newinstance{-\vec{\alpha}}{E}$.
\item $ \newcosts{-\vec{\alpha},E}{\newinstance{-\vec{\alpha}}{E}} = \costs{\instance} $.
\end{itemize}

Then, (b) follows.

The following statements are equivalent,
where the first statement follows from the second one by Observation~\ref{obs_settol}(b).
\begin{itemize}
\item There exists a solution $ S^\ast $ with $ E \subseteq \E \setminus S^\ast $
which is optimal for $ \instance $ and is \emph{not} optimal
for $ \newinstance{-\vec{\alpha}}{E}$.
\item There exists a solution which is optimal for $ \instance $ and is \emph{not} optimal
for $ \newinstance{-\vec{\alpha}}{E}$.
\item $ \newcosts{-\vec{\alpha},E}{\newinstance{-\vec{\alpha}}{E}} < \costs{\instance} $.
\end{itemize}

Then, (c) follows.

\item[\bf (a) ]
This follows directly from (b) and (c) by setting $E := \{ e \}$.
\hfill $\qed$
\end{description}
\end{proof}

Parts (b), (c), (d) of the following theorem show
that~\cite[Theorem 11]{GJM06b}
holds for each $e \in \E$ for the new single lower tolerance definition~\eqref{newdef_sin_low}.

\begin{theorem}
\label{thm_lt_newdef}
Let $ e \in \E $.
\begin{description}
\item[\bf (a)]
Let~$ e \in \E \setminus \lte $.
Then $ \nlt{\instance}{e} \in \{0,\infty\} $
holds,
where if $\costname$ is of type sum or product,
the lower tolerance is always $0$.

\item[\bf (b)]
$\nlt{\instance}{e}=\costs{\Dplus{e}}-\costs{\instance}$, if
$\costname$ is of type sum,
\item[\bf (c)]
$\nlt{\instance}{e}=
\frac{\costs{\Dplus{e}}-\costs{\instance}}{\costs{\Dplus{e}}}\cdot c(e)$,
if $\costname$ is of type product,
\item[\bf (d)]
$\nlt{\instance}{e}=\left\{\begin{array}{cl}
c(e)-c^\ast, & \mbox{if $g(e) < c^\ast $}\\
\infty, & \mbox{otherwise}
\end{array} \right.$,
where $ c^\ast := \costs{\instance} $ and
\begin{eqnarray*}
g(e) := \left\{
\begin{array}{cc}
\min_{S \in \Dplus{e}} \max_{a\in S \setminus \{e\}} \{ c(a) \},
& \mbox{if $ \Dplus{e} \not= \emptyset $}
\\
\infty, & \mbox{if $ \Dplus{e} = \emptyset $}
\end{array}
\right.,
\end{eqnarray*}
if $\costname$ is of type bottleneck.

\end{description}
\end{theorem}

\begin{proof}
For (a), the case $ e \in \lte $ does not occur.
For (b), (c), (d), the case $ e \in \lte $ is covered by~\cite[Theorem 11]{GJM06b},
Thus, in the following, let~$ e \in \E \setminus \lte $, i.e., $e$
lies inside each optimal solution.
\begin{description}
\item[(b)] $\costname$ is of type sum.

Let $ \alpha > 0 $ be arbitrary.
Then it holds that
$\newcosts{-\alpha,e}{\newinstance{-\alpha}{e}} = \costs{\instance} - \alpha $.
Thus, decreasing $ c(e) $ by $ \alpha $ decreases the objective value, and
$ \nlt{\instance}{e} = 0$ follows.

As $e$ lies in each optimal solution, it follows:
\begin{eqnarray*}
\nlt{\instance}{e} \; = \; 0 \; = \;
\costs{\Dplus{e}}-\costs{\instance}.
\end{eqnarray*}
Thus, (b) holds.

\item[(c)] $\costname$ is of type product.

Let $ \alpha $ be arbitrary with $ 0 < \alpha < \maxdec{e} $.
Then it holds that
$\newcosts{-\alpha,e}{\newinstance{-\alpha}{e}} = \costs{\instance} \cdot
\frac{c(e)-\alpha}{c(e)} $.
As all elements have positive cost, decreasing $ c(e) $ by $ \alpha $
decreases the objective value, and
$ \nlt{\instance}{e} = 0$ follows.

As $e$ lies in each optimal solution, it follows:
\begin{eqnarray*}
\nlt{\instance}{e}
\; = \; 0 \; = \;
\frac{\costs{\Dplus{e}}-\costs{\instance}}{\costs{\Dplus{e}}}\cdot c(e).
\end{eqnarray*}
Thus, (c) holds.

\item[(d)] $\costname$ is of type bottleneck.

Clearly, $ c(e) \le c^\ast~$ holds.
We distinguish between two cases:
\begin{description}
\item[{\bf i)}]
There does not exist an~$ \alpha > 0~$ with
$\newcosts{-\alpha,e}{\newinstance{-\alpha}{e}} < \costs{\instance} $.

It follows that~$ \nlt{\instance}{e} = \infty $.
It holds that $ g(e) \ge c^\ast$, and thus, the correctness of the assertion follows
in this case.

\item[{\bf ii)}]
There exists an~$ \alpha > 0~$ with
$\newcosts{-\alpha,e}{\newinstance{-\alpha}{e}} < \costs{\instance} $.

If~$c(e) < c^\ast$ holds, it follows that
$\newcosts{-\alpha,e}{\newinstance{-\alpha}{e}}= \costs{\instance} = c^\ast$
for all $\alpha > 0$, which conflicts with the assumption.
Therefore, it holds that $ c(e) = c^\ast~$ and
for all $ \alpha > 0~$ that
$\newcosts{-\alpha,e}{\newinstance{-\alpha}{e}} < \costs{\instance} $.
It follows that~$ \nlt{\instance}{e} = 0 $.

It holds that $ g(e) < c^\ast$. Because of that and because of
$ c(e)=c^\ast$, the correctness of the assertion follows in this case.

\end{description}
\item[(a)] This follows directly from the proofs of (b), (c), (d).
\hfill $\qed$
\end{description}
\end{proof}

Note that each single lower tolerance of a CSP and a CPP requires the solution of
a new instance. This instance can be formed by setting $c(e)$ to a sufficiently
small number~\cite[Theorem 10(a),(b)]{GJM06b}.  The computation of a
single lower tolerance of a CBP requires the solution of either an additional CSP or
an additional CBP instance~\cite[Algorithm 7, 8]{TJ22}.

If $\costname$ is of type bottleneck, we use the following observation.

\begin{observation}
\label{obs_bottl}
Let $ \instance =(\E,D,c,\costname)$ be a CBP instance
with optimal objective value $ c^\ast $ and $ E = \{e_1,e_2,\dots,e_k\} \subseteq \E $.
Let $ \vec{\alpha} = (\alpha_1,\alpha_2,\dots,\alpha_k) $.
Then the following holds:
\begin{eqnarray*}
\newcosts{-\vec{\alpha},E}{\newinstance{-\vec{\alpha}}{E}} & \ge &
\min \left\{ c^\ast, \min_{l=1}^k \left\{ c(e_l) - \alpha_l \right\} \right\}.
\end{eqnarray*}
\end{observation}

The following theorem is a generalization of \cite[Theorem 17]{JT18}
applied to the newly defined regular set lower tolerance~\eqref{newdef_reg_lowb} for
arbitrary $ E \subseteq \E $.

\begin{theorem}
\label{lrg_comp2}
Let $ \instance =(\E,D,c,\costname)$ be a CMP instance
and $ E = \{e_1,e_2,\dots,e_k\} \subseteq \E $.
Then the following holds:
\begin{description}
\item[{\bf (a)}]
$ \max_{l=1}^k \, \{ \nlt{\instance}{e_l} \} \; \le \;
\nlt{\instance}{E} \; \le \; \sum_{l=1}^k \nlt{\instance}{e_l} $.
\item[{\bf (b)}]
If $\costname$ is of type bottleneck,
$ \nlt{\instance}{E} \; = \;
\sum_{l=1}^k \nlt{\instance}{e_l} $.
\end{description}
\end{theorem}

\begin{proof}
\begin{description}
\item[\bf (a)]
We show the claimed inequalities.
\item[]
\hspace*{1em}
$ \max_{l=1}^k \, \{
\nlt{\instance}{e_l} \} \le \nlt{\instance}{E} $.

This follows directly from Remark~\ref{rem_newdef2}(a).

\item[]
\hspace*{1em}
$ \nlt{\instance}{E} \le
\sum_{l=1}^k \nlt{\instance}{e_l} $.

Trivially, the inequality is true
if $ \sum_{l=1}^k \nlt{\instance}{e_l} = $ $ \infty $.
In the following let
$ \sum_{l=1}^k \nlt{\instance}{e_l} $ $ \not=
\infty $.
Assume that $ \nlt{\instance}{E} >
\sum_{l=1}^k \nlt{\instance}{e_l} $.
Choose $ \alpha \in \R $ with
$ \sum_{l=1}^k
\nlt{\instance}{e_l} < \alpha \le
\nlt{\instance}{E} $,
$ \alpha = \sum_{l=1}^k
\alpha_l $,
$ \vec{\alpha} =
(\alpha_1,\alpha_2,\dots,
\alpha_k) $
and $ 0 \le
\alpha_l <
\maxdec{e_l} $
for $ l=1,2,\dots,k $
so that
$ \newcosts{-\vec{\alpha},E}{\newinstance{-\vec{\alpha}}{E}} \; = \;
\costs{\instance} $ holds.
Then a $ t \in \{1,2,\dots,k\} $ exists with $
\alpha_t
> \nlt{\instance}{e_t} $.
By the definition of $ \nlt{\instance}{e_t} $, namely \eqref{newdef_rev_lowb},
it holds:
\begin{eqnarray*}
\newcosts{-\vec{\alpha},E}{\newinstance{-\vec{\alpha}}{E}} \; \le \;
\newcosts{-\alpha_l,e_l}{\newinstance{-\alpha_l}{e_l}} \; < \;
\costs{\instance}.
\end{eqnarray*}
We receive a contradiction
to $ \newcosts{-\vec{\alpha},E}{\newinstance{-\vec{\alpha}}{E}} \; = \;
\costs{\instance} $.
The correctness of the assertion follows.

\item[{\bf (b)}]
By {\bf (a)}, it remains to be shown that $ \sum_{l=1}^k \nlt{\instance}{e_l} \le
\nlt{\instance}{E} $.

Trivially, the inequality is true if $ \nlt{\instance}{E} = \infty $.
If $ \sum_{l=1}^k \nlt{\instance}{e_l} = \infty $ holds, then
it also holds that $ \max_{l=1}^k \, \{ \nlt{\instance}{e_l} \} = \infty $.
From (a)
it follows that $ \nlt{\instance}{E} = \infty $, and the inequality is also
true. In the following let both terms
$ \nlt{\instance}{E} $ and $ \sum_{l=1}^k \nlt{\instance}{e_l} $
be not equal to $ \infty $.

Let $ \alpha \ge 0 $ with
$ \alpha = \sum_{l=1}^k \alpha_l,
\; \vec{\alpha} = (\alpha_1,\alpha_2,\dots,\alpha_k) $
and $ \alpha_l = \nlt{\instance}{e_l} $ for $ l = 1,2,\dots,k $.
By Theorem~\ref{thm_lt_newdef}(d), $  \alpha_l = \nlt{\instance}{e_l} =
c(e_l) - c^\ast  \not= \infty $ holds.
It follows:
\begin{eqnarray*}
\newcosts{-\vec{\alpha},E}{\newinstance{-\vec{\alpha}}{E}} & \ge &
\min \{ c^\ast, \min_{l=1}^k \{ c(e_l) - (c(e_l) - c^\ast) \} \}
\\
& = & c^\ast \; = \; \costs{\instance}.
\end{eqnarray*}
Note that the first inequality holds for each possible
$ \vec{\alpha} = (\alpha_1,\alpha_2,\dots,\alpha_k) $
by Observation~\ref{obs_bottl} and specifically for $  \alpha_l = c(e_l) - c^\ast $
for $ l = 1,2,\dots,k $.
Equality and the correctness of the assertion follow.
\hfill $\qed$
\end{description}
\end{proof}

The following theorem is a generalization of~\cite[Theorem 9]{JT24}
applied to the newly defined reverse set lower tolerance~\eqref{newdef_rev_lowb} for
arbitrary $ E \subseteq \E $.

\begin{theorem}
\label{lrv_comp2}
Let $ \instance =(\E,D,c,\costname)$ be a CMP instance
and $ E = \{e_1,e_2,\dots,e_k\} \subseteq \E $.
Then the following holds:
\begin{description}
\item[{\bf (a)}]
$ \nltb{\instance}{E} \; \le \;
\min_{l=1}^k \, \{ \nlt{\instance}{e_l} \}$.
\item[{\bf (b)}]
If $\costname$ is of type sum or product,
$ \nltb{\instance}{E} \; = \;
\min_{l=1}^k \, \{ \nlt{\instance}{e_l} \} $.
\end{description}
\end{theorem}

\pagebreak[3]

\begin{proof}
\begin{description}
\item[\bf (a)]
This follows directly from Remark~\ref{rem_newdef2}(b).

\item[\bf (b)]
By {\bf (a)} it remains to be shown  that
$ \min_{l=1}^k \, \{ \nlt{\instance}{e_l} \} \le \nltb{\instance}{E} $.

Trivially, the inequality is true if  $ \nltb{\instance}{E} = \infty $.
In the following let $ \nltb{\instance}{E} \not= \infty $.
Assume that $ \min_{l=1}^k \, \{ \nlt{\instance}{e_l} \} > \nltb{\instance}{E} $.

\pagebreak[3]

We distinguish between two cases:
\begin{description}
\item[(A)] There is at least one $ t \in \{1,2,\dots,k\} $ such that $e_t$ lies
in an optimal solution.

As for an objective function of type sum or product, decreasing the cost of an
element of an optimal solution by any $ \epsilon > 0 $
decreases the objective value, it follows:
$ 0 \; = \; \min_{l=1}^k \, \{ \nlt{\instance}{e_l} \}
\; > \; \nltb{\instance}{E} \; \ge \; 0 $.
This is a contradiction, and the correctness of the assertion follows.

\item[(B)] There is no $ t \in \{1,2,\dots,k\} $ such that $e_t$ lies in
an optimal solution.

It follows that $ \{e_1,e_2,\dots,e_k\} \in \lts $.
The correctness of the assertion follows by~\cite[Theorem 9(b)]{JT24}
and by Theorem~\ref{thm_cons}(c).
\hfill $\qed$
\end{description}
\end{description}
\end{proof}

\subsubsection*{Summary}

Our results show that the new definitions are consistent with the current ones
(see Theorem~\ref{thm_cons}). The computation formulas for single lower tolerances can be generalized
to $\E$ (see Theorem~\ref{thm_lt_newdef}(b),(c),(d)) and it does not require the set
$\lte$. Finally, the bounds and exact results for set lower tolerances
apply to all $E \subseteq \E$ (see Theorems~\ref{lrg_comp2}, \ref{lrv_comp2}).

\section*{Conclusions and future work}

In this work, we have addressed the challenge that current definitions of tolerances do not
apply to all elements or subsets of elements in combinatorial optimization problems. We
need to distinguish between upper tolerances, which concern cost increases of elements,
and lower tolerances, which concern cost decreases of elements.

The direct way to resolve this challenge is by extending the current definitions
to all elements. We use these extended definitions for upper tolerances,
in which the supremum
cost increases (of subsets) of elements are measured, but we find that current
computation formulas and bounds cannot be extended. In practice, it means that one first
has to determine whether elements belong to or are outside a given optimal solution
to compute upper tolerances.

For lower tolerances, on the other hand, we propose new definitions
based on the optimal objective value. We find that these new definitions are
consistent with the current definitions on its domain, and we find that the
computation formulas and bounds apply to all (subsets of) elements, meaning that we
no longer have to determine which elements belong to or are outside an optimal
solution for lower tolerances. Unfortunately, this approach cannot be applied to upper tolerances
in a meaningful way.

We hope that using this extended theory of set tolerances, heuristical
and exact algorithms will become more effective and easier to create.

\section*{Acknowledgement}

This work was supported by the Swedish Research Council grant 2022-04535.

\bibliographystyle{plain}
\bibliography{newdeftol37}

\end{document}